\documentclass[11pt]{article}
\usepackage{amssymb,amsfonts,amsmath, psfrag,eepic}
\usepackage{epsfig}
\usepackage{graphicx}
\newtheorem{theorem}{Theorem}[section]
\newtheorem{lemma}[theorem]{Lemma}

\newtheorem{corollary}[theorem]{Corollary}

\newenvironment{proof}[1][Proof]{\begin{trivlist}
\item[\hskip \labelsep {\bfseries #1}]}{\end{trivlist}}

\newenvironment{example}[1][Example]{\begin{trivlist}
\item[\hskip \labelsep {\bfseries #1}]}{\end{trivlist}}

\newcommand{\qed}{\nobreak \ifvmode \relax \else
      \ifdim\lastskip<1.5em \hskip-\lastskip
      \hskip1.5em plus0em minus0.5em \fi \nobreak
      \vrule height0.75em width0.5em depth0.25em\fi}

\title{Constraints for  Evolution Equations
with  Some Special Forms of Lax Pairs and
Distinguishing Lax Pairs by  Available Constraints
}
\author{YuQi Li $^{a}$ \footnote{E-mail: liyuqi@nbu.edu.cn},
Biao Li $^{a}$, SenYue Lou $^{a}$\\
\small $a$.\quad  Center for Nonlinear Science,  Ningbo University,
Ningbo 315211, China}

\begin{document}
\maketitle
\noindent {\large \bf Abstract:}
The constraints for evolution equations with some special forms of Lax pairs
are first investigated.
We show by examples how the method is rooted in the classical literatures
and how the ignored constraints provide nontrivial solutions.
Then we show, by the example of the KdV equation,
how this special form  of Lax pair may be found by the  method of Wahlquist-Estabrook.
At last we propose  how to impose constraints for  general Lax pairs
including nonlinear ones. With the proposition
the true Lax pairs and the fake ones can be distinguished easily.
The linearity nature in integrable partial differential equations seems to have been revealed.

\section{Introduction}
\hspace*{0.6cm}In general, the last step of the inverse scattering method
on the real line is to solve an integral equation such as
Gelfand-Levitan-Marchenko equation.
But only the reflection-less potentials, which are solitons, make up of
a closed system and therefore can be solved completely.
It had been observed that
the reflection-less potential is some function of the eigenfunctions.
With this observation the method of
nonlinearization \cite{CAOCW} of Lax pair was suggested.
Then it was observed \cite{LiYSh2}
that the constraint between the  potential
and the  eigenfunctions  may also be regarded as  symmetry constraint.
Today the method of symmetry constraint has
become a powerful tool for analyzing the solutions of  integrable systems.
By the method of nonlinearization of Lax pair or symmetry constraint,
we will obtain a wider class of solutions much more than the solitons
though the constraint maybe just be obtained by an observation
on the soliton and its corresponding  eigenfunctions.
In fact we will get algebraic-geometric solutions in most cases.
The algebraic-geometric solutions are  too complicated in practice,
especially when we are only interested in the numerical integration of
some initial-boundary problem for
an integrable partial differential equations (PDE).
Now it becomes more and more clear that
a high-precision numerical integration of an integrable PDE
should only cope with the  constraint
between the potential and the eigenfunctions.
Yet we may expect  those kinds of constraint will play a further role
in both analytical and numerical applications.

But the constraints are not so easy to find.
The reason is that the form of the constraints
may vary from one equation or hierarchy to another,
and moreover for the same equation
maybe there are several types of constraints,
which may lead to completely different types of solutions.
Now the method of symmetry constraint is still popular to get  constraints,though a lot of useful constraints may be lost by it.
In the first part of this paper we will analyze a wide class of
Lax pairs, from which the constraints will arise naturally.
The constraints provided here may be or not be the symmetry constraint.
Therefore, sometimes the class of solutions may be expanded.
It must be point out that the special form of Lax pairs
can also be gotten by the Wahlquist-Estabrook (WE) \cite{EW} method.

  Sakovich \cite{Sokovich} had pointed out
by example that Lax pairs with nontrivial
spectral parameter may be fake ones.
In fact the fake  Lax pairs have long been puzzling.
Kaup\cite{KAUP} had advised  a postulate to distinguish the fake Lax pairs
from the true ones. But his postulate is not algorithmic.
And  it also seems to be a little narrow   in practice.
Furthermore there are nonlinear Lax pairs
though they may be linearized in most known examples.
An algorithm for distinguishing the fake Lax pairs
from the true ones should apply in nonlinear Lax pairs
as well as the  usual ones.
In this paper we propose  a  functional equation,
a solution to which will enable us to provide
a nontrivial constraint for the integrable PDE.
The constraint is simply a superposition
of the  solutions to the functional equation.
Probably we may guess
this superposition reflects more or less the linear nature of integrable PDEs.

The paper is organized as follow.
Section 2  propose a kind of natural constraints for
a special form of Lax pairs.
Several  theorems there  will be proved
to guarantee the justice of the  proposal.
In addition several examples are provide in the section,
where some interesting results may be found.
Section 3  first shows  by the example of the KdV equation
how the WE method is capable  of getting  the special form of
Lax pairs described in Section 2.
Then a general requirement for constraining
Lax pairs including nonlinear Lax pairs
is proposed. At last the proposal is illustrated by
two nonlinear Lax pairs of the KdV equation.
Section 4 summarizes the main results.

\section{Constraints for evolution equations with a special form of Lax pair}
In the following paper, in order to avoid confusion by symbols, we define various symbols as follows: $f_t=\dot f=\frac{\partial f}{\partial t}$;
$f_x=f'=\frac{\partial f}{\partial x}$;
$f^{(n)}=\frac{\partial^n f}{\partial x^n} $;
$\hat P 1$ is the function
obtained by applying operator $\hat P$ to function $f=1$;
Pseudo-operator $\partial^{-1}$ is
the inverse of operator $\partial$ \cite{DICKY};
$g[u]$ denotes a differential polynomial of $u$.

\subsection{Observations about constraints for PDEs}
\hspace*{0.6cm}Numerical experiments have shown that most of
the numerical integrations of initial boundary value problems (IBVPs)
for PDEs are much more difficult  than the IBVPs for ordinary differential
equations(ODE).
But for some PDEs the integrations may be done
by  only solving a series of ODEs.
This celebrated  property should be considered as
a kind of integrability for the PDE.
As an example let us explain how the periodic KdV equation $u_t=6 u u_x +u_{xxx}$
is integrable in this sense. The stationary KdV equation is $f_n=0$, where
$f_n=\hat L^n \, u_x$ and $\hat L=\partial^2 +4 u+2 u_x \partial^{-1}$.
It is well-known that $f_n$ is a (2n+1)-order
differential polynomial of $u$.
So $f_n=0$ may be written into a system of
first-order ODEs for $2n+1$ variables $u,u^{(1)},\cdots,u^{(2n)}$.
Also we must know the time evolution for $u,u^{(1)},\cdots,u^{(2n)}$.
In fact
\begin{eqnarray}
\frac{d}{dt} u^{(i)}= \frac{d^i}{dx^i} (6 u u^{(1)}+u^{(3)}),
\quad \mbox{Mod} \quad f_n=0, \quad i=0,1,\cdots, 2n. \label{TevolutionKdV}
\end{eqnarray}
So we have 2 sets of ODEs for the KdV equation. One set governs the
space evolution for $u,u^{(1)},\cdots,u^{(2n)}$. And the other
governs the time evolution for $u,u^{(1)},\cdots,u^{(2n)}$.
So for any $(x,t)$ the values of  $u,u^{(1)},\cdots,u^{(2n)}$ may be
gotten from the initial values
$u_0,u_0^{(1)},\cdots,u_0^{(2n)}$ at $(x_0, t_0)$
by a space evolution and a time evolution.
In other words for some special kind of initial value problems the
periodic KdV equation is solvable by only integrating some ODEs.
It should be noted that the variables $u,u^{(1)},\cdots,u^{(2n)}$
are only 'natural' variables, not  'canonical' variables.
The ODEs written for the 'canonical' variables are much simpler, see the
following sections.

Clearly from the point of view of numerical integration,
the only crucial fact for the integration of the periodic KdV equation
is that $f_n=0$ is invariant under the KdV  flow
or in other words $f_n=0$ is an invariant manifold (IM) of the KdV equation.
But for a given PDE its IMs are not so easy to give, especially
when the required IMs must have some completeness.
Luckily nonlinearization of Lax pair or symmetry constraint just provide a way
to provide IMs for PDEs. Most popularly the symmetry constraint
makes use of  Lax pair in the inverse scattering transformation
\begin{eqnarray}
&&\hat L \psi_i=\lambda_i \psi_i , \label{LIST}\\
&&\frac{\partial}{\partial t} \psi_i=\hat P \psi_i . \label{PIST}
\end{eqnarray}
Still take the KdV equation as an example.
The well-known nonlinearization of Lax pair
\begin{eqnarray}
&&\hat L \psi_i=(\partial^2+u) \psi_i=\lambda_i \psi_i , \label{lIST}\\
&&\frac{\partial}{\partial t} \psi_i=\hat P \psi_i=(4 \partial^3+
3 (u \partial + \partial u) ) \psi_i, \label{pIST}
\end{eqnarray}
of the KdV equation is \begin{eqnarray}
u=c_0+\sum_{i=1}^n c_i \psi_i^2, \label{KdVCTClassic}
\end{eqnarray}
where $\psi_i=\psi_i(x,t)$. With the constraint (\ref{KdVCTClassic}),
Equation (\ref{lIST}) and Equation (\ref{pIST}) become
two sets of ODEs.

A less popular representation of Lax equation is
\begin{eqnarray}
&&f_{i+1}=\hat L f_i , \label{fndef}\\
&&\frac{\partial}{\partial t} f_i=\hat P f_i , \label{MovEq}
\end{eqnarray}
for any $i \in N$.
It is easy to verify that the 	compatibility condition for
(\ref{fndef}) and (\ref{MovEq}) is also the Lax equation in its operator form
\begin{eqnarray}
\frac{d }{d t} \hat L= [\hat P,\hat L] .\label{LaxEq}
\end{eqnarray}
So Equations (\ref{fndef}) and  (\ref{MovEq}) is equivalent to
Equations (\ref{lIST}) and (\ref{pIST}).
Sometimes a constraint for (\ref{fndef}) and (\ref{MovEq}) is more convenient.

There are several noticeable aspects for constraining the Lax equations.

Our first observation is that if $\hat P$ is a differential operator then
$f_n(x,0)=0$ is equivalent to $f_n(x,t)=0$.
Let us take the KdV equation as an example to explain why this hold.
The well-known $\hat P$ of the KdV equation is
$\hat P=4 \partial^3+3 (u \partial + \partial u)$.
So $f_n$ satisfies $f_t=4 f_{xxx}+6 u f_x+3 u_x f$.
By analyzing the Taylor series of $f_n$
we know that if $f_n$ is analytic and $f_n(x,0)=0$, then $f_n(x,t)=0$.
Therefore, if $\hat P$ is a differential operator
then (\ref{fndef}) and (\ref{MovEq}) can be truncated by $f_n=0$.
For the same reason the truncation may also be made by setting
$f_{n}=\sum_{i=1}^{n-1} a_i f_i$.

Our second observation is that the number of equations in
the truncated (\ref{fndef}) and (\ref{MovEq}) is very close to
but not enough
for (\ref{fndef}) and (\ref{MovEq}) to be  2 sets of ODEs.
Also take the  KdV equation as an example.
To make it easy we take its natural Lax pair
$\hat L=\partial^2+4 u+2 u_x \partial^{-1}$,
$\hat P=\partial(\partial^2+6u)$.
The  number of equations in
the truncated (\ref{fndef}) and (\ref{MovEq}) including the KdV equation and
the truncation condition is
$2(n-1)-1+1+1=2n -1$. But we need $2(n-1)+2=2 n$
for $f_1,f_2 ,\cdots,f_{n-1}$ and $u$
in (\ref{fndef}) and (\ref{MovEq}) to be ODE systems.
So we need an additional relation between $u$ and $f_i$.
This is just the nonlinearization of Lax pairs.
For the example of the KdV equation
the additional relation may be $f_1=u_x$,
which is just the usual way to get the stationary solutions.

Our third observation is that for the periodic KdV equation
the IMs mentioned above  is dense.
But for some other systems or another Lax pair,
the relevant IMs may be not dense.
For a nontrivial constraint of Lax equations
it should be the minimum requirement
that the number of the IMs   increases
as the truncation number $n$ increases.

The above three  requirements rule out quite a lot of fake Lax pairs.
But some true Lax pairs may also be ruled out.
For example the famous short-pulse (SP) equation \cite{SchaferWayne}
$u_{xt}=u+u u_x^2+1/2 u^2 u_{xx}$
has a true Lax pair \cite{Brunelli}
\begin{eqnarray}
\hat L&=&(\partial^{-1}+u_x \partial^{-1} u_x) \partial^{-1}, \label{SPl} \\
\hat P&=&\partial^{-1}+\frac{1}{2} \partial u^2 .\label{SPp}
\end{eqnarray}
Here $\hat P$ is not a differential operator. And furthermore SP  equation
is  not an evolution equation.
So the Lax pair (\ref{SPl}) and (\ref{SPp}) does not fulfill  our  first requirement.
Thereafter we will constrain ourselves in evolution equations.

\subsection{A special form of Lax pair and the relevant natural constraint}

\hspace*{0.6cm}We have mentioned that
with the usual Lax pair (\ref{lIST}), (\ref{pIST})
the KdV equation has  constraint (\ref{KdVCTClassic}).
Except this famous constraint, we find that the KdV equation
with the less popular Lax pair (\ref{fndef}) , (\ref{MovEq})
also has another interesting constraint
\begin{eqnarray}
u=c_0+ \sum_{i,j=1}^n c_{i+j} f_i f_j, \label{kdvCons2}
\end{eqnarray}
where $f_k=0$ for $k>n$ and $c_k=0$ for $k<n+1$.
So for any $n \in Z$, $f_n=0$ provides an IM for the KdV equation.

At the first sight
constraint (\ref{KdVCTClassic})  or  (\ref{kdvCons2})
is  strange. In this section
we will investigate a special form of
Lax pairs and give a natural constraint for them.
Then constraints (\ref{KdVCTClassic}) and  (\ref{kdvCons2})
can be understood directly and easily.

We start by the following simple lemma.
\begin{lemma}
If $\hat P $ is  differential operator and
 $g[u]$ satisfies  $\frac{d}{dt} g[u]=\hat P g[u] $,
then  constraint  $g[u]=\sum_{i=1}^n a_i \psi_i$ or  constraint
$g[u]=\sum_{i=1}^n a_i f_i$ ,
 $\sum_{i=1}^n b_i f_i=0$ exist,
where $\psi_i$, $f_i$ are defined in (\ref{LIST}),(\ref{PIST}),
(\ref{fndef}), (\ref{MovEq}).
\end{lemma}
The proof is obvious.

For a given  Lax pair, it is not a easy task to find out
a nontrivial constraint.
However, there exist a natural constraint
for Lax pairs  in the following form (Form I):
\begin{itemize}
\item $\hat L=\hat L^+ + L^F \partial^{-1}$,
where $L^+$ is a differential operator and $L^F$ is a function;
\item $\hat P$ is a differential operator.
\end{itemize}
\begin{theorem} \label{theo01}
For a Lax pair in Form I,
one possible constraint is $L^F=\sum_{i=1}^n a_i \psi_i$,
where $\hat L \psi_i=\lambda_i \psi_i$.
And the other constraint   $\sum_{i=0}^n a_i f_i=0$,
where $f_i=\hat L^i \, L^F$ is also possible.
\end{theorem}
\begin{proof}
The first part.
We should only prove $\frac{\partial}{\partial t} L^F=\hat P  L^F$,
because we already have
$\frac{\partial}{\partial t} \psi_i=\hat P  \psi_i$.
The Lax equation is $\dot {\hat L}=[\hat P,\hat L]$.
The negative part of the Lax equation is
$\dot L^F \partial^{-1}=(\hat P L^F) \partial^{-1}-
L^F \partial^{-1} (\bar P \, 1)$,
where $\bar P$ is the operator conjugate of $\hat P$.
Then we get $\partial \frac{\dot L^F-\hat P L^F}{L^F}=-(\bar P \, 1) \partial$.
So $\frac{\dot L^F-\hat P L^F}{L^F}=-(\bar P \, 1)$ and
$(\frac{\dot L^F-\hat P L^F}{L^F})_x=0$. So $\bar P \, 1$ must be a constant.
But we can redefine $\hat P$ by adding some constant to it
such that  $\bar P \, 1=0$. Then we have $\dot L^F-\hat P L^F=0$.

The second part. We will verify $f_i$ satisfies (\ref{MovEq}).
$f_0$ satisfies (\ref{MovEq}), which has been proven in the first part.
Suppose $f_i$  satisfies (\ref{MovEq}).
Then $\frac{d}{dt} f_{i+1}=\frac{d}{dt}(\hat L f_i)
=(\frac{d}{dt} \hat L) f_i+  \hat L (\frac{d}{dt} f_i )=
(\hat P \hat L-\hat L \hat P) f_i+ \hat L \hat P f_i=
\hat P \hat L f_i=\hat P f_{i+1}$.
\end{proof}

From the second part of the proof we  know
it is a useful method to nonlinearize the Lax system
$\{ (\ref{fndef}) , (\ref{MovEq}) \}$ by finding out  a differential
polynomial $f_0$ such that $\frac{d}{dt} f_0=\hat P f_0$.
But this is not the only way, for example, see constraint (\ref{kdvCons2}).

\subsection{Some  examples}

\hspace*{0.6cm}The following three examples are all Lax equations in Form I.
The first two examples, which are both of the KdV equation,
show that different forms of Lax pair for the same equation
may naturally lead to  different kinds of solutions.
The  third example, which is  the Ito equation,
shows how a special form of Lax pair may be used by two ways
to solve the same equation.

\begin{example}
The recursion operator of the KdV equation

It is well-known that the recursion operator and
the linearization operator form a natural Lax pair. Particularly to
the KdV equation we have
\begin{eqnarray}
&&\hat L=\partial^2 +4 u+2 u_x \partial^{-1} , \label{newKdVL}\\
&&\hat P=\partial (\partial^2 +6 u). \label{newKdVP}
\end{eqnarray}
It is also fairly easy to verify that the Lax equation with
the above $\hat L$ (\ref{newKdVL}) and $\hat P$ (\ref{newKdVP})
is the KdV equation $u_t=6 u u_x +u_{xxx}$.
Obviously $\hat L$ (\ref{newKdVL}) and $\hat P$ (\ref{newKdVP}) is in Form I.
Here $L^F=2 u_x$. So a constraint $u_x=\sum a_i \psi_i$ exists, where
$\psi_i$ satisfies $(\partial^2 +4 u+2 u_x \partial^{-1}) \psi_i=
\lambda_i \psi_i $ and $\frac{\partial}{\partial t}\psi_i=\psi_{ixxx}+
6 (u \psi_i)_x$.
Another constraint is $\sum c_i f_i=0$, where
$f_0=u_x$, the famous finite gap constraint.

This gives a simple explanation for why (\ref{KdVCTClassic}) holds.
Note this classical constraint is the integration of symmetry constraint.
Let $\phi_i=\frac{d}{dx} \psi_i^2$, where $\psi_i$ satisfies
(\ref{lIST}) and (\ref{pIST}). It is easy to verify
\begin{eqnarray}
&&(\partial^2 +4 u+2 u_x \partial^{-1}) \phi_i =4 \lambda_i \phi_i,
\label{kdvphiL}\\
&&\frac{d}{dt} \phi_i=\partial (\partial^2 +6 u) \phi_i \label{kdvphiP}
\end{eqnarray}
So $u_x=\sum_{i=1}^n a_i \phi_i=\sum_{i=1}^n a_i  \frac{d}{dx}\psi_i^2$
i.e., $u= c_0+\sum_{i=1}^n a_i \psi_i^2 $ is a proper
constraint.

How about (\ref{kdvCons2})? Let $f_j=4^j \sum_{i=j}^n \psi_i \psi_{n+j-i}$.
Then we can verify $\hat L f_j=f_{j+1}$ and
$\frac{\partial}{\partial t} f_j =\hat P f_j$,
where
\begin{eqnarray}
&&\hat L=\partial^2+4 u-2 \partial^{-1}u_x, \label{conjL}\\
&&\hat P=\partial^3+ 6 u \partial. \label{conjP}
\end{eqnarray}
It is obvious $\frac{\partial}{\partial t}(u-c_0)=\hat P (u-c_0)$.
So $u=c_0+\sum a_i f_i$ is a proper constraint, which is equivalent
to (\ref{kdvCons2}).

\end{example}

\begin{example} The soliton Lax pair of the KdV equation

It is known that the KdV  equation still has another Lax pair
\begin{eqnarray}
\hat L&=&\partial + u \partial^{-1}, \label{kdvL2}\\
\hat P&=&\partial^3+ 3 \partial u, \label{kdvP2}
\end{eqnarray}
which is also in Form I.
So $u=\sum_{i=1}^n a_i \psi_i$ or $u=\sum_{i=1}^n a_i f_i$
is a proper constraint.

We are extremely interested in the case $\hat L^n 0=0$.
\end{example}
\begin{theorem}
 $(\partial + u \partial^{-1})^{n+1} 0=0$, $n=1,2,\cdots$,
generates all  $n$-soliton solitons of the KdV equation.
\end{theorem}
\begin{proof}
Recall the soliton solutions for the KdV equation\cite{LiYSh} is
\begin{eqnarray}
&&N_l=\exp (-2 \kappa_l x+ \theta_l)
  \left(1+\sum_{j=1}^n \frac{c_j N_j}{\kappa_l+
  \kappa_j} \right) , \quad l=1,2,\cdots, n\nonumber \\
&&u=2 \sum_{j=1}^n c_j N_j' . \nonumber
\end{eqnarray}
First we will prove
\begin{eqnarray}
N_{l}''+u N_l=-2 \kappa_l N_{l}'.\nonumber
\end{eqnarray}
Or
\begin{eqnarray}
\hat L N_{l}'=-2 \kappa_l N_{l}'+\gamma_{1l} u. \nonumber
\end{eqnarray}
Here the appearance of $\gamma_{1l} u$ is due to $\partial^{-1}$, which is
an indefinite integral operator.
Then
\begin{eqnarray}
\hat L^m u=2 \sum_{j=1}^n \left[ (-2 \kappa_j)^m +
    \gamma_{1} (-2 \kappa_j)^{m-1}+\cdots+
    \gamma_{m} ) c_j N_{j}'\right] . \nonumber
\end{eqnarray}
From the expression of $\hat L^n u=0$ we known that $\kappa_j$ is
completely determined by $\gamma_j$:
\begin{eqnarray}
(-2 \kappa_j)^n +
    \gamma_{1} (-2 \kappa_j)^{n-1}+\cdots+
    \gamma_{m} =0.
\end{eqnarray}
Altogether the $n$-solitons satisfy $\hat L^n 0=0$.
But $n$-soliton solutions have $2 n$ free parameters
and $\hat L^n 0=0$ can be written as  $2 n$ first-order ODEs.
So the general solution of $\hat L^n 0=0$ is the $n$-soliton solutions.
\end{proof}

\begin{example} Ito equation (Drinfeld-Sokolov II)

The Ito equation \cite{DS, Ito}
\begin{eqnarray}
&&u_t=3 v_x , \label{NewEqMov1}\\
&&v_t=(u v)_x+v_{xxx} \label{NewEqMov2}
\end{eqnarray}
has a Lax pair \cite{DS}
\begin{eqnarray}
&&\hat L=\partial^3 +u\partial +u_x+v\partial^{-1},
\label{NewEqL} \\
&&\hat P=\partial (\partial^2 +u ), \label{NewEqP}
\end{eqnarray}
which is in Form I.
We will investigate the constraints $\hat L^n 0=0$.
The first  few constraints of
(\ref{NewEqMov1}) and (\ref{NewEqMov2}) can be easily solved.

The first constraint is $\hat L 0=v=0$. By (\ref{NewEqMov1}) $u_t=0$.
So $u=u_0(x)$.

The second constraint is $\hat L^2 0=v_{xxx}+(u v)_x+ v \int v dx=0$.
By (\ref{NewEqMov2}) $v_t=-v \int v dx$.
Let us introduce a new variable $\phi$ such that $\phi_x=v$.
Then $\phi_{xt}=-\phi_x \phi=-\frac{1}{2}(\phi^2)_x$.
So $\phi_t=-\frac{1}{2} \phi^2+g(t)$.
We can prove $g(t)$ is independent on $t$.
So
\begin{eqnarray}
\phi_t =-\frac{1}{2} \phi^2+\frac{c_0}{2}. \label{inva01}
\end{eqnarray}
The solution of (\ref{inva01}) is $\phi=\sqrt{c_0}
 \tanh (\frac{\sqrt{c_0}}{2}t+c_1)$.
$c_0$ is independent on $x$. But $c_1$ may be dependent on $x$.
So the final result of $\psi$ is
\begin{eqnarray}
\phi=\sqrt{c_0} \tanh (\frac{\sqrt{c_0}}{2}t+c_1(x) ). \label{solu01}
\end{eqnarray}
We may prove
\begin{eqnarray}
\phi_{xxx}+u \phi_x +\frac{1}{2} \phi^2=\frac{c_0}{2}. \label{constrpsi01}
\end{eqnarray}
So
\begin{eqnarray}
&&u=\frac{1}{\phi_x}(\frac{c_0}{2}-\frac{1}{2} \phi^2-\phi_{xxx}),
\label{soluu01}\\
&&v=\phi_x \label{solv01}.
\end{eqnarray}
is one solution of (\ref{NewEqMov1}) and (\ref{NewEqMov2}).

The third constraint is $\hat L^3 0=0$.
By a tedious study of this constraint,
we find if  $\phi$ satisfies the following equation
\begin{eqnarray}
\phi_t=-1/2 \phi^2+f_1(t) \phi +f_2 , \nonumber
\end{eqnarray}
where
\begin{eqnarray}
\frac{d f_1}{dt}=\gamma(t) , \quad
\frac{d \gamma}{dt} =f_1(t) \gamma(t) , \nonumber
\end{eqnarray}
and $f_2$ is an arbitrary constant,
then $v$  and $u$ expressed  by
\begin{eqnarray}
v&=&\phi_x , \nonumber\\
u&=&\frac{1}{\phi_x} (\phi_t -\phi_{xxx}-\gamma ). \label{ItoSolution}
\end{eqnarray}
is a solution.

It seems very difficult  to solve $\hat L^n 0=0$, $n \geq 3$.
So how to analyze the high-mode vibrations?

First by Theorem \ref{theo01} we can impose a
constraint $v=\sum_{i=1}^n a_i \psi_i$, where
$\hat L \psi_i=\lambda_i \psi_i$ and $\psi_{it}=\hat P \psi_i$.
After a  little modification of the form of the Lax equation, we get
\begin{eqnarray}
&&
\left(\frac{\lambda_i \psi_{i}-\psi_i'''-(u \psi_i)'}{v}
\right)'=\psi_i , \nonumber\\
&&\dot \psi_i=\psi_i'''+(u \psi_i)' , \nonumber \\
&&\dot u=3 v'. \nonumber
\end{eqnarray}
Let $\psi_i=\varphi_i'$. Then $v=\sum_{i=1}^n a_i \varphi_i'$. So we get
\begin{eqnarray}
&&\lambda_i \varphi_{i}'-\varphi_i''''-(u \varphi_i')'=
 v \varphi_i , \label{intg01} \\
&&\dot \varphi_i=\varphi_i'''+u \varphi_i' , \label{git01}\\
&&\dot u=3 v' . \label{utvx01}
\end{eqnarray}
Multiplying (\ref{intg01}) by $a_i$ and summing over $i$ we get
\begin{eqnarray}
\sum_{i=1}^n a_i (\lambda_i \varphi_{i}'-\varphi_i''''-(u \varphi_i')' )
= v \sum_{i=1}^n a_i \varphi_i
=\frac{1}{2} \left((\sum_{i=1}^n a_i \varphi_i)^2 \right)' . \label{intg02}
\end{eqnarray}
Integrating (\ref{intg02}) we immediately  get
\begin{eqnarray}
\sum_{i=1}^n a_i (\lambda_i \varphi_{i}-\varphi_i'''- u \varphi_i' )=
\frac{1}{2} \left(\sum_{i=1}^n a_i \varphi_i \right)^2+\gamma. \nonumber
\end{eqnarray}
At  first glance $\gamma$ is a function of $t$.
But by (\ref{utvx01}) we known $\dot \gamma(t)=0$. So $\gamma$ is a constant.
Now $u$ is solved as
\begin{eqnarray}
u=-\frac{1}{v}\left[\frac{1}{2} (\sum_{i=1}^n a_i \varphi_i)^2+v''+\gamma
-\sum_{i=1}^n a_i \lambda_i \varphi_{i}\right] . \label{ItoExpressU}
\end{eqnarray}
Totally Equation (\ref{intg01}) and Equation (\ref{git01})
contain $(n-1)+n=2n-1$
differential equations. So it is not enough for them
to be ODE systems.
The usual way to overcome this problem is  to introduce  another constraint.
But for the Ito equation there is no need to do such things.
The reason is as follows.
Obviously we have
\begin{eqnarray}
\dot \varphi'_i=\lambda_i \varphi'-\varphi_i \sum_{j=1}^n a_j \varphi'_j.
\label{phixtIto}
\end{eqnarray}
By differentiate (\ref{phixtIto}) with respect to $x$  once and twice
we get another 2 sets of differential equations.
Together with (\ref{git01}) and (\ref{phixtIto})
we have a closed form for the $t$-evolution of
$\varphi_i$, $\varphi'_i$, $\varphi''_i$  and $\varphi'''_i$.
This is  enough for a numerical computation for the Ito equation.
Note that Equation (\ref{intg01}) may be regarded as  $4 n-1$ first-order
differential constraints in the $x$-direction,
i.e., there is still one arbitrary function $\varphi_{i_0}(x)$
among the $4 n$ initial functions $\varphi_i$.
\end{example}

\subsection{Further generalizations}
\hspace*{0.6cm}Now  we will generalize Form I a little. Anyway we
have the following theorem:
\begin{theorem}  \label{theo02}
If $\hat L=\hat L^+ +\sum_{i=1}^n f_i \partial^{-1} g_i$, where
$\hat L^+$ is a differential operator, and also $\hat P$ is
a differential operator, then there a linear composition
of $f_i$s
$L^F= \sum_{j=1}^n \alpha_i f_i$ such that
$\frac{d}{dt} L^F=\hat P L^F$, i.e.,
a nonlinearization $L^F=\sum_{i=1}^m \beta_i \psi_i$ exists,
where $\psi$ is eigenfunction $\hat L \psi=\lambda \psi$.
\end{theorem}
We say  $\hat L$ and $\hat P$ in Theorem \ref{theo02} is of Form II.
To prove Theorem \ref{theo02} we need the following
Lemma \ref{lemma01} and Corollary \ref{cor01}.

\begin{lemma} \label{lemma01}
Iff $\sum_{i=1}^n f_i(x) \partial^{-1} g_i(x)=0$,
then $\sum_{i=1}^n f_i(x) g_i(y)=0$.
\end{lemma}
\begin{proof}
$\sum_{i=1}^n f_i(x) \partial^{-1} g_i(x)=0$ is
equivalent to $\sum_{i=1}^n f_i(x) g_i^{(k)}(x)=0$, $k=0,1,2,\cdots$.
So $\sum_{i=1}^n f_i(x) g_i(y)=\sum_{i=1}^n f_i(x) g_i(x+(y-x))=
\sum_{i=1}^n \sum_{k=0}^\infty\frac{1}{k!} f_i(x) g_i^{(k)}(x) (y-x)^k=0$.
\end{proof}
\begin{corollary} \label{cor01}
If $\sum_{i=1}^n A_i(x) \partial^{-1} B_i(x)=
\sum_{i=1}^n F_i(x) \partial^{-1} G_i(x)$  and  both $\{A_i(x)$\}
and \{ $G_i(x)\}$
are linearly independent, then $B_i(x)$ is a linear composition of
$G_j(x)$ and $F_i(x)$ is a linear composition of
$A_j(x)$.
\end{corollary}

Theorem \ref{theo02} guarantees that there is a natural constraint
for a Lax pair in Form II.
For Lax pairs in  more complicated form, such as $\hat L$ and $\hat P$
are matrix, we have not reach a general result like Theorem \ref{theo01}.
But it seems always a good guess
that a linear composition $v$ of the components
before $\partial^{-1}$ satisfies $\dot v=\hat P v$.

\begin{example}ZS-AKNS

With the help of recursion operator the AKNS hierarchy can
be expressed as a simple expression
\begin{eqnarray}
&&\hat L=\frac{1}{i} \left(
\begin{array}{cc}
-\partial+2 q \partial^{-1} r& 2 q \partial^{-1} q\\
-2 r \partial^{-1}r & \partial-2 r \partial^{-1} q
\end{array}
\right), \label{AKNSL}\\
&& \left(
\begin{array}{c}
q\\r
\end{array}
\right)_t=\hat L^n
\left( \begin{array}{c}
-i q\\i r
\end{array}
\right) ,  \label{AKNSeq}
\end{eqnarray}
where $i$ is the imaginary unit.
Here $\hat L$ is the recursion operator.
So $\hat L$ and the linearization operator $\hat P$ form
a natural Lax equation $\frac{d}{dt}{\hat L}=[\hat P, \hat L]$.
The symmetry $\sigma_i=\left(\begin{array}{c}\zeta_i \\ \xi_i\end{array}
\right)$ satisfies
\begin{eqnarray}
&&\sigma_{i+1}=\hat L \sigma_i, \nonumber\\
&&\frac{d}{dt} \sigma_i=\hat P \sigma_i. \nonumber
\end{eqnarray}
But $\left(\begin{array}{c}q\\ -r \end{array} \right)$ is
also a symmetry because
if $\left(\begin{array}{c}q\\ r \end{array} \right)$ is a solution of
(\ref{AKNSL}) and (\ref{AKNSeq}) then
$\left(\begin{array}{c}\bar q\\\bar r \end{array}\right)=
\left(\begin{array}{c}k q\\\ r/k \end{array}\right)$
 is also a solution, where $k \in \mathbb{R}$.
So we  know
\begin{eqnarray}
\frac{d}{dt} \left( \begin{array}{c} q\\-r \end{array} \right)=
\hat P  \left( \begin{array}{c} q\\-r \end{array} \right) . \nonumber
\end{eqnarray}
Therefore one possible constraint is
$\displaystyle \left( \begin{array}{c}q\\-r
\end{array}\right)=\sum_{i=1}^m a_i \sigma_i$.
Meanwhile we can also write the Lax equation as
$\hat L \Psi=\lambda \Psi$ and $\frac{d}{dt} \Psi=\hat P \Psi$, where
$\Psi$ is a $2\times1$ vector.
So $\displaystyle
\left( \begin{array}{c} q\\-r \end{array} \right) =\sum_{i=1}^n a_i \Psi_i$
is also one possible constraint.

For example AKNS $n=3$ is the coupled KdV equation
\begin{eqnarray}
 \left(
\begin{array}{c}
q\\r
\end{array}
\right)_t=
\left( \begin{array}{c}
6 q r q_x-q_{xxx}\\6 q r r_x-r_{xxx}
\end{array}
\right). \nonumber
\end{eqnarray}
The corresponding linearization operator is
\begin{eqnarray}
\hat P=\left(
\begin{array}{cc}
6 r \partial q-\partial^3& 6 q q_x\\
6 r r_x& 6 q \partial r-\partial^3
\end{array}
\right). \label{CoupledKdVP}
\end{eqnarray}
It is easy to verify
\begin{eqnarray}
 \left(
\begin{array}{c}
q\\-r
\end{array}
\right)_t=\hat P
\left( \begin{array}{c}
q\\-r
\end{array}
\right). \nonumber
\end{eqnarray}
Then one possible constraint is
\begin{eqnarray}
 \left(
\begin{array}{c}
q\\-r
\end{array}
\right)=\sum_{i=1}^n  a_i \left(
\begin{array}{c}
\psi_i\\\phi_i
\end{array}
\right), \nonumber
\end{eqnarray}
where $\left(\begin{array}{c}\psi_i\\ \phi_i\end{array}\right)=\Psi_i$
satisfies
$\hat L \Psi_i=\lambda_i \Psi_i$, $\frac{d}{dt} \Psi_i=\hat P \Psi_i$
and the expressions for $\hat L$ and $\hat P$ are
(\ref{AKNSL}) and (\ref{CoupledKdVP}).
\end{example}

\subsection{The equivalence of the two kinds of constraint}
\hspace*{0.6cm}We have introduced two kinds of constraints
in Lemma \ref{lemma01} and thereafter.
Now we will  explain the relation between the two kinds of constraint.

Recall the constraint of the first kind is
\begin{eqnarray}
g[u]=\sum_{i=1}^n a_i \psi_i . \label{conskind1}
\end{eqnarray}
Constraint of the second kind is
\begin{eqnarray}
g[u]=\sum_{i=1}^n b_i f_i . \label{conskind2}
\end{eqnarray}
In general $\psi_i$ and $f_i$ have not any relation.
But with constraint (\ref{conskind1}) or (\ref{conskind2})
they become linear dependent in most cases.
Suppose $\displaystyle f_i=\sum_{j=1}^n c_{ij} \psi_j$.
Then from the relation $f_{i+1}=\hat L f_i$, we get
\begin{eqnarray}
\sum_j c_{i+1,j} \psi_j=\sum_j \lambda_j c_{i j} \psi_j. \nonumber
\end{eqnarray}
So $c_{i+1,j}=\lambda_j c_{ij}$.

\begin{theorem}
For a constraint of the first kind
$\displaystyle g[u]=\sum_{i=1}^n a_i \psi_i$
satisfying that if $i\neq j$ then $\lambda_i \neq \lambda_j$,
there exist a  constraint of the second kind
$\displaystyle g[u]=\sum_{i=1}^n b_i f_i$,
$\displaystyle f_{n+1}=\sum_{i=1}^n k_i f_i$
where $b_i$ and $k_i$ is determined by  $c_{i+1,j}=\lambda_j c_{ij}$,
$\displaystyle a_j=\sum_{i=1}^n b_i c_{ij}$,
$\lambda_i^n =k_j \lambda_i^{j-1}$ with
arbitrary $c_{1i} \neq 0$.
\end{theorem}
\begin{proof}Only to consider that if $\prod c_{1i} \neq 0 $ and
 $\lambda_i \neq \lambda_j$ then $\det (c_{ij}) \neq 0$. \end{proof}

\begin{theorem}
For a constraint of the second kind
$\displaystyle g[u]=\sum_{i=1}^n b_i f_i$,
$\displaystyle f_{n+1}=\sum_{i=1}^n k_i f_i$ satisfying that
$\displaystyle \lambda^n =\sum_{j=1}^n k_j \lambda^{j-1}$
 has distinct root for $\lambda$,
there exist a  constraint of the first kind
$\displaystyle g[u]=\sum_{i=1}^n a_i \psi_i$,
where $\lambda_i$ and $a_i$ is determined by:
$\lambda_i$ is the $i$-th root of
$\displaystyle \lambda^n =\sum_{j=1}^n k_j \lambda^{j-1}$ ,
 $c_{1i} \neq 0$,
$c_{i+1,j}=\lambda_j c_{ij}$,
$\displaystyle a_i=\sum_j b_j c_{ji}$, .
\end{theorem}

\section{Searching the  special form of Lax pairs and
constraining the nonlinear Lax pair systems}

\hspace*{0.6cm}It is well-known that the WE method is
a powerful tool for  searching the Lax pairs for PDEs.
By  the preceding examples we have known that one PDE can
have completely different Lax pairs. So which Lax pair
may be found by the WE method is heavily dependent on the original assumptions
for the form of  Lax pair.
Furthermore the Lax pairs obtained by the WE method are  often
nonlinear ones, which must be linearized before further application.
So very often the nonlinear Lax pairs are considered to be useless.
In this section we will first demonstrate, by the example of the KdV equation,
 how to search the Lax pair in  Form II.
Then we will propose a functional equation for finding constraints
for both linear or nonlinear Lax pair systems.

Let us first demonstrate by the example of KdV equation
how Lax pair (\ref{kdvL2}), (\ref{kdvP2})
can be found by the usual WE method.
The KdV equation is first written into
\begin{eqnarray}
u_x&=&z, \nonumber\\
z_x&=&p, \nonumber\\
u_t&=&6 u z+p_x . \label{kdv1ord}
\end{eqnarray}
Substituting (\ref{kdv1ord}) to the zero-curvature equation
\begin{eqnarray}
M_t-N_x+[M,N]=0, \nonumber
\end{eqnarray}
where $M=M(u,z,p)$,$N=N(u,z,p)$,
we immediately get
\begin{eqnarray}
z_t M_z+ p_t M_p +p_x (M_u-N_p)+6 u z M_u-
z N_u-p N_z+[M,N]=0. \nonumber%\label{zr01}
\end{eqnarray}
Here $z_t$, $p_t$ and $p_x$ should be considered as independent variables.

{\sl Remark:} The KdV equation in the form of (\ref{kdv1ord}), in fact,
has implied that $M$ is only dependent on $u$. In other words,
we are searching Lax pairs for the KdV equation
whose $M$ is only dependent on $u$.

Then by the usual steps of the WE method we get
\begin{eqnarray}
M&=&X_1+u X_2+u^2 X_3, \nonumber \\
N&=&X_4+(3 u^2+p) X_2+(4 u^3+2 u p-z^2 )X_3+z [X_1,X_2]+u [X_1,[X_1,X_2]]+
\frac{u^2}{2} [X_2,[X_1,X_2]], \nonumber
\end{eqnarray}
where
\begin{eqnarray}
[X_1,X_3]& = & 0, \nonumber \cr
[ X_2 , X_3 ] & = & 0 , \nonumber \cr
[X_1,X_4]&=&0, \nonumber \cr
[X_1,[X_1,[X_1,X_2]]]+[X_2,X_4]&=&0, \nonumber\cr
[X_2,[X_2,[X_1,X_2]]]&=&0, \nonumber \\
\frac{1}{2} [X_1,[X_2,[X_1,X_2]]]+ 3 [X_1,X_2]+
[X_2,[X_1,[X_1,X_2]]]+[X_3,X_4] &=&0 . \label{Xcomm}
\end{eqnarray}
We would first not consider (\ref{Xcomm}) as an open Lie structure,
but rather a matrix equation. Then we will get linear Lax pair.
The $2 \times 2$ realization of (\ref{Xcomm}) can be solved by
Maple
and there are 12 solutions  obtained by Maple.
However, only one solution is worthy of notice.
The solution is
\begin{eqnarray}
&&x_{3;{2,1}}=0,\cr
&&x_{2;{1,2}}=0,\cr
&&x_{2;{2 ,2}}=x_{2;{1,1}},\cr
&&x_{3;{1,2}}=0,\cr
&&x_{3;{2,2}}=x_{3;{1,1}},\cr
&&x_{2;{2,1}}=-
{x_{1;{1,2}}}^{-1},\cr
&&x_{4;{2,1}}=x_{1;{2,1}} \left( {x_{1;{1,1}}}^{2}+{x
_{1;{2,2}}}^{2}-2\,x_{1;{1,1}}x_{1;{2,2}}+4\,x_{1;{1,2}}x_{1;{2,1}}
 \right) , \cr
&&x_{4;{1,2}}={x_{1;{1,1}}}^{2}x_{1;{1,2}}+x_{1;{1,2}}{x_{1;{2
,2}}}^{2}-2\,x_{1;{1,1}}x_{1;{1,2}}x_{1;{2,2}}+4\,{x_{1;{1,2}}}^{2}x_{
1;{2,1}}, \cr
&&x_{4;{2,2}}=x_{4;{1,1}}-{x_{1;{1,1}}}^{3}-3\,x_{1;{1,1}}{x_{1
;{2,2}}}^{2}+3\,{x_{1;{1,1}}}^{2}x_{1;{2,2}}-4\,x_{1;{1,1}}x_{1;{1,2}}
x_{1;{2,1}}\cr
&&\hspace*{1.3cm}+{x_{1;{2,2}}}^{3}+4\,x_{1;{1,2}}x_{1;{2,1}}x_{1;{2,2}},
\label{sl11}
\end{eqnarray}
where $x_{i;j,k}$ is the $(j,k)$ element of matrix  $X_i$.
Thanks to (\ref{sl11}),  the equation
\begin{eqnarray}
\Psi_x=M \Psi, \label{psix}
\end{eqnarray}
where $\Psi=\left(\begin{array}{c}\psi(x,t)\\\phi(x,t)
\end{array}\right)$ is a $2$-dimensional vector, is
completely determined.
The goal  equation is $\hat L \varphi=\lambda \varphi$.
$\varphi$ may be the linear composition of $\psi$ and $\phi$
\begin{eqnarray}
\varphi=f_1([u]) \psi+f_2([u]) \phi . \label{vphdef}
\end{eqnarray}
Clearly $\varphi$ satisfies a $2$-order ODE.
If we demand the
$2$-order ODE do not contain $u^{(n)}$,$n \geq 1$, then
we obtain $f_2=0$, $f_1=Const$, $x_{3;1,1}=0$ and
$x_{2;1,1}=0$. Obviously, we may set $f_1=1$.
Now (\ref{sl11}) is reduced to
\begin{eqnarray}
\varphi_{xx} -(x_{1;1,1}+x_{1;2,2} )\varphi_x+
(u- x_{1;1, 2} x_{1;2, 1}+ x_{1;1, 1} x_{1;2, 2}) \varphi=0. \label{rdkdvL}
\end{eqnarray}
Comparing Equation (\ref{rdkdvL}) with the goal equation
$\hat L \varphi=\lambda \varphi$, we immediately get
\begin{eqnarray}
\varphi_{xx} +\lambda_1\varphi_x+
(u+\lambda_2) \varphi=0. \label{rdkdvL2}
\end{eqnarray}
To write a Lax pair in Form II from (\ref{rdkdvL2}),
we only need to introduce the variable $\xi=\varphi_x$.
Now (\ref{rdkdvL2}) can be rewritten to
\begin{eqnarray}
\hat L \xi \equiv (\partial+(u+\lambda_2)\partial^{-1})\xi=-\lambda_1 \xi.
\nonumber%\label{rdkdvL3}
\end{eqnarray}
Now $\hat L$ here is in Form II.
That $\hat P$ is also in Form II can be verified directly.

With regard to the other 11 solutions of (\ref{sl11}),
four of them can not be reduced to $2$-order ODE;
seven of them can not
be reduced to $2$-order ODE whose coefficients are only functions of $u$.

Now we  turn to nonlinear Lax pairs.
It has NEVER been demanded  the Lax pair  be linear.
Then where is the linearity?
We find the following functional equation  is crucial in integrable systems.
The functional equation  is
\begin{eqnarray}
\frac{d}{dt} g[u,\psi,\lambda]=\hat A[u] g[u,\psi,\lambda],
\label{CrucialLinear}
\end{eqnarray}
where $g[u,\psi,\lambda]$ denotes a function of finite variables
including  the original variables  $u^{(i)}$,
auxiliary variables  $\psi_i$ and   complex parameters $\lambda$.
But $\hat A[u]$ is  a linear differential operator which only
involves functions of finite variables of  $u^{(i)}$.
Of course, non-degenerate conditions
$\frac{\delta}{\delta u} g[u,\psi,\lambda] \neq 0$ and
$\frac{\delta}{\delta \psi} g[u,\psi,\lambda] \neq 0$ must be satisfied.
The final constraint for the general Lax system is
\begin{eqnarray}
\sum_\lambda g[u,\psi,\lambda]=0. \label{CrucialConstraint}
\end{eqnarray}
It seems that
the  linearity  nature of  Lax integrable systems is completely characterized
by Equation (\ref{CrucialLinear}) and (\ref{CrucialConstraint}).
To see how these arguments works, let us still take the KdV equation
as an example.

Following WE we regard Equation (\ref{Xcomm}) as a Lie algebra constraint.
Here we only cite the
final equations for the pseudo-potentials:
\begin{eqnarray}
&&\frac{\partial y_2}{\partial x}=-e^{2 y_3}, \hspace{1.2cm}
\frac{\partial y_2}{\partial t}=-2 e^{2 y_3} (u+2 \lambda),\nonumber\\
&&\frac{\partial y_3}{\partial x}=-y_8, \hspace{1.5cm}
\frac{\partial y_3}{\partial t}=z- 2 (u+2 \lambda) y_8, \nonumber\\
&&\frac{\partial y_8}{\partial x}=\lambda-u-y_8^2, \quad
\frac{\partial y_8}{\partial t}=2 z y_8-2 (u+2 \lambda)(u+y_8^2-\lambda)-p.
 \label{EWnlLax}
\end{eqnarray}
Note  the coefficients have been adjusted to fit the KdV equation here.
Because (\ref{EWnlLax}) is a nonlinear Lax pair, the subsequent step
is  usually to  find some transformation to linearize (\ref{EWnlLax}).
But Equation (\ref{CrucialLinear}) and (\ref{CrucialConstraint}) enable
us to solve the KdV equation straightly.
Our task is to find a function $g$ and
a linear operator $\hat A$ satisfying (\ref{CrucialLinear}).
It is easy to check
\begin{eqnarray}
g=\alpha e^{-2 y_3}+\beta u+\gamma, \quad \hat A=\partial^3+6 u \partial,
\quad \alpha,\beta,\gamma \in {\mathbb C}, \label{gAsolu01}
\end{eqnarray}
is a solution of (\ref{CrucialLinear}).
So
\begin{eqnarray}
\sum_{i=1}^n (\alpha_i e^{-2 y_{3i}} +\beta_i u+\gamma_i)=
\gamma+\beta u+\sum_{i=1}^n  \alpha_i e^{-2 y_{3i}}=0, \label{KdVpsi2Con}
\end{eqnarray}
is a proper constraint, where $y_{3i}$ satisfies
(\ref {EWnlLax}) with $\lambda=\lambda_i$.

Given a nonlinear Lax pair, (\ref{CrucialLinear}) can always be
proposed no matter whether the nonlinear Lax pair is linearizable.
The difficulty is  to give   nontrivial $g$ and $\hat A$.

Open Lie algebra (\ref{Xcomm}) has more than one solutions.
Let us assume $X_3=0$. (It seems also to be not clear yet
when $X_3$ is a centre.)  One 4-dimensional realization of (\ref{Xcomm})
is
\begin{eqnarray}
&&[X_1,X_2]=-\alpha \beta X_4+\frac{1}{\alpha^2}X_5,
[X_1,X_4]=0,  [X_1,X_5]=2 \alpha^2 X_1+\alpha X_2, \nonumber\\
&&[X_2,X_4]=\beta X_4-\frac{1}{\alpha^3}X_5,
[X_2,X_5]=\alpha^2 \beta^2 X_1-2 \alpha^2 X_2-\beta X_5,
[X_4,X_5]=2 \alpha X_1 +X_2.
\nonumber
\end{eqnarray}
If we set $\bar X_4=X_4-\frac{1}{\alpha} X_1$, then $\bar X_4$ commute with
any other vectors.
If there is no centre then $\bar X_4=0$, i.e., $X_4=\frac{1}{\alpha}X_1$.
Therefore, in fact we have a 3-dimensional realization  of (\ref{Xcomm})
\begin{eqnarray}
[X_1,X_2]=\frac{1}{\alpha^2} X_5-\beta X_1,
[X_1,X_5]=2 \alpha^2 X_1+\alpha X_2,
[X_2,X_5]=\alpha^2 \beta^2 X_1-2 \alpha^2 X_2-\beta X_5.
\label{comuDIM3}
\end{eqnarray}
Following  WE  we get the following nonlinear Lax pair
\begin{eqnarray}
&&\frac{\partial y_1}{\partial x}=-1+\alpha u(2-\frac{\alpha}{\gamma} e^{
 \frac{y_1}{\sqrt{\alpha}}}-\frac{\gamma}{\alpha} e^{
 -\frac{y_1}{\sqrt{\alpha}}}), \nonumber\\
&&\frac{\partial y_1}{\partial t}=-\frac{1}{\alpha}+4 \alpha u^2+
2 \alpha u''-\frac{\alpha}{\gamma}(u+2 \alpha u^2+
\sqrt{\alpha} u'+\alpha u'') e^{ \frac{y_1}{\sqrt{\alpha}}} \nonumber\\
&&\hspace{1.2cm}-\frac{\gamma}{\alpha} (u+2 \alpha u^2-
 \sqrt{\alpha} u'+\alpha u'')
 e^{-\frac{y_1}{\sqrt{\alpha}}}, \nonumber\\
&&\frac{\partial y_2}{\partial x}=u(\frac{\alpha^2}{\gamma^2} e^{
 \frac{y_1}{\sqrt{\alpha}}}-e^{ -\frac{y_1}{\sqrt{\alpha}}}), \nonumber\\
&&\frac{\partial y_2}{\partial t}=\frac{\alpha^2}{\gamma^2}(
\frac{u}{\alpha}+2 u^2+\frac{u'}{\sqrt{\alpha}}+u'') e^{
 \frac{y_1}{\sqrt{\alpha}}}-(
\frac{u}{\alpha}+2 u^2-\frac{u'}{\sqrt{\alpha}}+u'') e^{
 -\frac{y_1}{\sqrt{\alpha}}} . \label{nlLax301}
\end{eqnarray}
By a definition of $\bar y_1=\frac{y_1}{\sqrt{\alpha}} +\ln \alpha-\ln \gamma$
and $\bar y_2=\frac{2 \alpha}{\gamma} y_2$, we can simplify
(\ref{nlLax301}) to
\begin{eqnarray}
&&\frac{\partial \bar y_1}{\partial x}=-\frac{1}{\sqrt{\alpha}}+
 2 \sqrt{\alpha} u- 2 \sqrt{\alpha} u \cosh(\bar y_1), \nonumber\\
&&\frac{\partial \bar y_1}{\partial t}=
\sqrt{\alpha}(-\frac{1}{\alpha^2}+ 4  u^2+2  u'')
-2 \sqrt{\alpha}(\frac{u}{\alpha}+2  u^2+u'') \cosh(\bar y_1)
-2  u' \sinh (\bar y_1), \nonumber\\
&&\frac{\partial \bar y_2}{\partial x}=
u \sinh (\bar y_1), \nonumber\\
&&\frac{\partial \bar y_2}{\partial t}=
\frac{u'}{\sqrt{\alpha}}  \cosh(\bar y_1)
+(\frac{u}{\alpha}+2 u^2+ u'')  \sinh (\bar y_1).
\label{nlLax302}
\end{eqnarray}
Still another substitution $\tilde y_1=e^{\bar y_1}$ ,
$\epsilon=\sqrt{\alpha} $ and $\tilde y_2=2 y_2$
transforms (\ref{nlLax302}) to a form more suitable for computation,
\begin{eqnarray}
&&\frac{\partial  \tilde y_1}{\partial x}=
-\epsilon u+ (2 \epsilon u-\frac{1}{\epsilon}) \tilde y_1
-\epsilon u \tilde y_1^2,  \nonumber\\
&&\frac{\partial  \tilde y_1}{\partial t}=
-\frac{u}{\epsilon}-2 \epsilon u^2+u'-\epsilon u''-
\epsilon(\frac{1}{\epsilon^4}-4 u^2-2 u'') \tilde y_1
-(\frac{u}{\epsilon}+2 \epsilon u^2+u'+\epsilon u'')  \tilde y_1^2,
\nonumber\\
&&\frac{\partial \tilde y_2}{\partial x}=
u \tilde y_1 -u \frac{1}{\tilde y_1} ,
\nonumber\\
&&\frac{\partial \tilde y_2}{\partial t}=
(\frac{u}{\epsilon^2}+2 u^2+\frac{u'}{\epsilon}+u'')\tilde y_1
-(\frac{u}{\epsilon^2}+2 u^2-\frac{u'}{\epsilon}+u'')
\frac{1}{\tilde y_1}. \label{nlLax303}
\end{eqnarray}
Now the task is to give $g$ and $\hat A$.
With the help of Maple we get the following solution
\begin{eqnarray}
&&g=c_0 +c_1 u+c_2 \frac{(\tilde y_1-1)^2}{\tilde y_1}
e^{\epsilon \tilde y_2}, \nonumber\\
&&\hat A=\partial^3+6 u\partial . \label{gAsolution}
\end{eqnarray}
Then it becomes very clear that (\ref{nlLax303}) and (\ref{EWnlLax}) must
be equivalent.
In fact, the transformation from (\ref{nlLax303}) to (\ref{EWnlLax}) is
\begin{eqnarray}
y_8=\frac{1}{2 \epsilon}\frac{1+\tilde y_1}{1-\tilde y_1},
\quad y_3=\frac{1}{2}\epsilon \tilde y_2 +
\ln (\tilde y_1-1)-\frac{1}{2}\ln \tilde y_1 ,
\quad \lambda=\frac{1}{4 \epsilon^2}.
\label{convertNLNL}
\end{eqnarray}

\section{Conclusions}
\hspace*{0.6cm}We have proposed a kind of  natural constraints
for evolution PDEs with a special kind of Lax pairs.
By the method several examples have been studied in detail.
At least two interesting results should be  noticed.
One is that KdV  equation has  a  soliton-Lax pair,
from which only soliton  solutions can appear.
The other  is that the Ito equation can be  constraint
to a series of PDEs that  are solvable by the method of characteristics.

For a given  Lax pair, it is always very difficult to solve
the functional equation (\ref{CrucialLinear}) completely.
Even for the KdV  equation
only special solutions of (\ref{CrucialLinear}) have
been obtained.
It is necessary to point out that solutions (\ref{gAsolu01})  and (\ref{gAsolution}) are both
obtained by the classical separation of variables.
More general ansatz  about the solutions of (\ref{CrucialLinear}) will
lead to  a set of too complicated equations to solve.
This also explains
why we seek the special kind of Lax pairs, of which the constraints are
manifest.
Equations (\ref{CrucialLinear}) and (\ref{CrucialConstraint})
can act as a measure of integrability
for evolution equations with Lax pairs.
Any generalization of (\ref{CrucialLinear}) and (\ref{CrucialConstraint})
seems so difficult.

\section*{Acknowledgments}
\hspace*{0.6cm}The first author wish to thank
Prof. \fbox{H. Y. Guo}, Prof. S. K. Wang  and Dr. D. S. Wang
for their helpful discussions and encouragements. The work
is partly supported by  NSFC (No.10735030), NSF of Zhejiang Province (R609077, Y6090592),
NSF of Ningbo City (2009B21003, 2010A610103, 2010A610095).

\end{document}